\documentclass{llncs}

\usepackage{fullpage}

\usepackage{amssymb,amsmath,stmaryrd}
\usepackage{pstricks, pst-node, pst-tree, pst-plot}
\usepackage{multirow}
\usepackage{slashbox, subfigure}
\usepackage{epsfig, amssymb, amsmath}
\usepackage{color, epsfig, epic, eepic}

\usepackage{algorithm}
\usepackage{algorithmic}
\def\N{{\mathbb{N}}}

\def\RP{{\mathbb{R_+}}}

\def\opt{{\mathrm{O}\textsc{pt}}}
\def\c{{T}}
\def\mypi{{\pi}}

\def\last{{\mathrm{last}}}
\def\l{{t}}

\allowdisplaybreaks

\begin{document}

\title{On the Impact of Fair Best Response Dynamics\thanks{This research was partially supported by the grant NRF-RF2009-08 ``Algorithmic aspects of coalitional games'' and by the PRIN 2008 research project COGENT (COmputational and GamE-theoretic aspects of uncoordinated NeTworks), funded by the Italian Ministry of University and Research.}}

\author{Angelo Fanelli\inst{1} Luca Moscardelli\inst{2} \and Alexander Skopalik\inst{1}}

\institute{Division of Mathematical Sciences, School of Physical and Mathematical Sciences, \\Nanyang Technological University, Singapore. \\Email: {\tt \{angelo.fanelli,ASkopalik\}@ntu.edu.sg} \and Department of Science, University of Chieti-Pescara, Italy. \\Email: {\tt moscardelli@sci.unich.it} 
}

\maketitle

\begin{abstract}
In this work we completely characterize how the frequency with which each player participates in the game dynamics affects the possibility of reaching efficient states, i.e., states with an approximation ratio within a constant factor from the price of anarchy, within a polynomially bounded number of best responses. We focus on the well known class of congestion games and we show that, if each player is allowed to play at least once and at most $\beta$ times any $T$ best responses, states with approximation ratio $O(\beta)$ times the price of anarchy are reached after $T \lceil \log \log n \rceil$ best responses, and that such a bound is essentially tight also after exponentially many ones.
One important consequence of our result is that the fairness among players is a necessary and sufficient condition for guaranteeing a fast convergence to efficient states. This answers the important question of the maximum order of $\beta$ needed to fast obtain efficient states, left open by \cite{FanelliFM08,FanelliM09} and \cite{AwerbuchAEMS08}, in which fast convergence for constant $\beta$ and very slow convergence for $\beta=O(n)$ have been shown, respectively.
Finally, we show that the structure of the game implicitly affects its performances. In particular, we show that in the symmetric setting, in which all players share the same set of strategies, the game always converges to an efficient state after a polynomial number of best responses, regardless of the frequency each player moves with.
\end{abstract}


\centerline{{\bf Keywords}: Congestion Games, Speed of Convergence, Best Response Dynamics.}

\section{Introduction}
The class of congestion games is a well established approach for modelling any non-cooperative system in which a set of resources are shared among a set of selfish players. In a congestion game we have a set of $m$ resources and a set of $n$ players. Each player's strategy consists of a subset of resources. The delay of a particular resource $e$ depends on its congestion, corresponding to the number of players choosing $e$, and the cost of each player $i$ is the sum of the delays associated with the resources selected by $i$. In this work we focus on linear congestion games where the delays are linear functions. A congestion game is called symmetric if all players share the same strategy set. A state of the game is any combination of strategies for the players and its social cost, defined as the sum of the players' costs, denotes its quality from a global perspective. The social optimum denotes the minimum possible social cost among all the states of the game.

{\bf Related work.}
Rosenthal \cite{R73} has shown, by a potential function argument, that the natural decentralized mechanism known as Nash dynamics consisting in a sequence of moves in which at each one some player switches its strategy to a better alternative, is guaranteed to converge to a pure Nash equilibrium \cite{N50}. 

In order to measure the degradation of social welfare due to the selfish behavior of the players, Koutsoupias and Papadimitriou \cite{KP99}  defined the price of anarchy as the worst-case ratio between the social cost in a Nash equilibrium and that of a social optimum. 
The price of anarchy for congestion games has been investigated by Awerbuch et al. \cite{AwerbuchAE05}  and Christodoulou and Koutsoupias \cite{ChristodoulouK05}. They both proved that the price of anarchy for congestion games with linear delays is $5/2$.

The existence of a potential function relates the class of congestion games to the class of polynomial local search problems (PLS) \cite{Johnson88}.
Fabrikant et al. \cite{FabrikantPT04} proved that, even for symmetric congestion games, the problem of computing Nash equilibria is PLS-complete \cite{Johnson88}. One major consequence of the completeness result is the existence of  congestion games with initial states such that any improvement sequence starting from these states needs an exponential number of steps in the number of players $n$ in order to reach a Nash equilibrium. A recent result by Ackermann et al. \cite{AckermannRV08} shows that  the previous negative result holds even in the restricted case of congestion games with linear delay functions.

The negative results on computing equilibria in congestion games have lead to the development of the concept of $\epsilon$-Nash equilibrium, in which no player can decrease its cost by  a factor of more than $\epsilon$. Unfortunately, as showed by  Skopalik and V{\"o}cking \cite{SkopalikV08}, also the problem of finding an $\epsilon$-Nash equilibrium in congestion games is PLS-complete for any $\epsilon$, though, under some restrictions on the delay functions, Chien and Sinclair \cite{ChienS07} proved that in symmetric congestion games the convergence to $\epsilon$-Nash equilibrium is polynomial in the description of the game and the minimal number of steps within each player has chance to move. 

Since negative results tend to dominate the issues relative to equilibria, another natural arising question is whether efficient states (with a social cost comparable to the one of any Nash equilibrium) can be reached by best response moves in a reasonable amount of time 
(see \cite{AwerbuchAEMS08,ChristodoulouMS06,FanelliFM08,FanelliM09}). We measure the efficiency of a state by the ratio among its cost and the optimal one, and we refers to it as the approximation ratio of the state. We generally say that a state is efficient when its approximation ratio is within a constant factor from the price of anarchy.
Since the price of anarchy of linear congestion games is known to be constant \cite{AwerbuchAE05,ChristodoulouK05}, efficient states approximate the social optimum by a constant factor.
It is worth noticing that in the worse case, a generic Nash dynamics starting from an arbitrary state could never reach a state with an approximation ratio lower than the price of anarchy. Furthermore, by a potential function argument it is easy to show that in a linear congestion game, once a state $S$ with a social cost $C(S)$ is reached, even if such a state is not a Nash equilibrium, we are guaranteed that for any subsequent state $S'$ of the dynamics, $C(S')= O(C(S))$.
\\ Awerbuch et al. \cite{AwerbuchAEMS08} have proved that for linear congestion games, sequences of moves reducing the cost of each player by at least a factor of $\epsilon$, converge to efficient states in a number of moves polynomial in $1/\epsilon$ and the number of players, under the minimal liveness condition that every player moves at least once every a polynomial number of moves. 
Under the same liveness condition, they also proved that exact best response dynamics  can guarantee the convergence to efficient states only after an exponential number of best responses \cite{AwerbuchAEMS08}.
Nevertheless, Fanelli et al. \cite{FanelliFM08} have shown that, under more restrictive condition  that each player plays exactly once every $n$ best responses, any best response dynamics converges to an efficient state after $\Theta( n \log\log n)$ best responses. Subsequently, Fanelli and Moscardelli \cite{FanelliM09} extended the previous results to the more general case in which each player plays a constant number of times every $O(n)$ best responses.


{\bf Our Contribution.}
In this work we completely characterize how the frequency with which each player participates in the game dynamics affects the possibility of reaching efficient states.
In particular, we close the most important open problem raised by  \cite{AwerbuchAEMS08} and  \cite{FanelliFM08,FanelliM09} for linear congestion games. On the one hand, in \cite{AwerbuchAEMS08} it is shown that, even after an exponential number of best responses, states with a very high approximation ratio, namely $\Omega\left(\frac{\sqrt n}{\log n}\right)$, can be reached. On the other hand, in \cite{FanelliFM08,FanelliM09} it is shown that, under the minimal liveness condition in which every player moves at least once every $T$ steps, if players perform best responses such that each player is allowed to play at most $\beta=O(1)$ times any $T$ steps (notice that $\beta=O(1)$ implies $T=O(n)$), after $T \lceil \log \log n \rceil$ best responses a state with a constant factor approximation ratio is reached.

The more $\beta$ increases, the less the dynamics is fair with respect to the chance every player has of performing a best response: $\beta$ measures the degree of unfairness of the dynamics.
The important left open question was that of determining the maximum order of $\beta$ needed to obtain fast convergence to efficient states: We answer this question by proving that, after $T \lceil \log \log n \rceil$ best responses, the dynamics reaches states with an approximation ratio of $O(\beta)$. Such a result is essentially tight since we are also able to show that, for any $\epsilon  >0$, there exist congestion games for which, even for an exponential number of best responses, states with an approximation ratio of $\Omega(\beta^{1-\epsilon})$ are obtained. Therefore, $\beta$ constant as assumed in \cite{FanelliFM08,FanelliM09} is not only sufficient, but also necessary in order to reach efficient states after a polynomial number of best responses.

Finally, in the special case of symmetric congestion games, we show that the unfairness in best response dynamics does not affect the fast convergence to efficient states; namely, we prove that, for any $\beta$, after $T \lceil \log \log n \rceil$ best responses efficient states are always reached.

The paper is organized as follows: In the next section provide the basic notation and definitions. Section \ref{sec:asymmetric} is devoted to the study of asymmetric congestion games, while Section \ref{sec:symmetric} to the analysis of the symmetric case.
Finally, Section \ref{sec:conclusion} extends the results to more general settings and gives some conclusive remarks.


\section{Model and Definitions}

A \emph{congestion game}
${\cal G} = \left(N, E, (\Sigma_i)_{i \in N}, (f_e)_{e \in E}, (c_i)_{i \in N}\right)$ is a non-cooperative strategic game characterized by the existence of a set $E$ of resources to be shared by $n$ players in $N=\{1,\ldots,n\}$.

Any strategy $s_i \in \Sigma_i$ of player $i \in N$ is a subset of resources, i.e., $\Sigma_i \subseteq 2^E$. A congestion game is called \emph{symmetric} if all players share the same set of strategies, i.e., $\Sigma = \Sigma_i$ for every $i \in N$. Given a state or strategy profile $S = (s_1, \ldots, s_n)$ and a resource $e$, the number of players using $e$  in $S$, called the congestion on $e$, is denoted by $n_e(S)=|\{i \in N \; | \; e \in s_i \}|$. A delay function $f_e : \N \mapsto \RP$ associates to resource $e$ a delay depending on the number of players currently using $e$, so that the cost of player $i$ for the pure strategy $s_i$ is given by the sum of the delays associated with the resources in $s_i$, i.e., $c_i(S)=\sum_{e \in s_i}f_e(n_e(S))$.

In this paper we will focus on linear congestion games, that is having linear delay functions with nonnegative coefficients. More precisely, for every resource $e \in E$, $f_e(x) = a_ex + b_e$ for every resource $e \in E$, with $a_e, b_e \geq 0$.

Given the strategy profile $S = (s_1, \ldots, s_n)$, the social cost $C(S)$ of a given state $S$ is defined as the sum of all the players' costs, i.e., $C(S) = \sum_{i \in N} c_i(S)$. An optimal strategy profile $S^* = (s^*_1, \ldots, s^*_n)$ is one having minimum social cost; we denote $C(S^*)$ by $\opt$.
The \emph{approximation ratio} of state $S$ is given by the ratio between the social cost of $S$ and the social optimum, i.e., $\frac{C(S)}{\opt}$.
Moreover, given the strategy profile $S = (s_1, s_2,\ldots, s_n)$ and a strategy $s'_i \in \Sigma_i$, let $(S_{-i}, s'_i) = (s_1, s_2,\ldots, s_{i-1}, s_{i}', s_{i+1}, \ldots, s_n)$, i.e., the strategy profile obtained from $S$ if player $i$ changes its strategy from $s_i$ to $s'_i$.

The potential function is defined as $\Phi(S) = \sum_{e \in E}\sum_{j = 1}^{n_e(S)} f_e(j)$. It is call \emph{exact} potential function since it satisfies the property that for each player $i$ and each strategy $s'_i \in \Sigma_i$ of $i$ in $S$, it holds that $c_i(S_{-i},s'_i) - c_i(S) = \Phi(S_{-i},s'_i) - \Phi(S)$.
It is worth noticing that in linear congestion games, for any state $S$, it holds $\Phi(S) \leq C(S) \leq 2 \Phi(S)$.

Each player acts selfishly and aims at choosing the strategy decreasing its cost, given the strategy choices of other players. 
A \emph{best response} of player $i$ in $S$ is a strategy $s^b_i \in \Sigma_i$ yielding the minimum possible cost, given the strategy choices of the other players, i.e., $c_i(S_{-i}, s^b_i) \leq c_i(S_{-i}, s'_i)$ for any other strategy $s'_i \in \Sigma_i$. Moreover, if no $s'_i \in \Sigma_i$ is such that $c_i(S_{-i}, s'_i) < c_i(S)$, the best response of $i$ in $S$ is $s_i$. 
We call a \emph{best response dynamics} any sequence of best responses.

Given a best response dynamics starting from an arbitrary
state, we are interested in the social cost of its final state.
To this aim, we must consider dynamics in which each player moves at least once every a given number $T$ of best responses, otherwise one or more players could be ``locked out'' for arbitrarily long and we could not expect to bound the social cost of the state reached at the end of the dynamics.
Therefore, we define a $T$-\emph{covering} as a dynamics of $T$ consecutive best responses in which each player moves at least once.
More precisely, a $T$-covering $R=\left(S^0_R,\ldots,S^{\c}_R\right)$ is composed by $\c$ best responses; $S^0_R$ is said to be the \emph{initial} state of $R$ and $S^{\c}_R$ is its \emph{final} state. For every $1 \leq \l \leq \c$, let $\mypi_{R}(\l)$ be the player performing the $\l$-th best response of $R$; $\mypi_{R}$ is such that every player performs at least a best response
in $R$. In particular, for every $1 \leq \l \leq \c$, $S^{\l}_R=\left(({S^{\l-1}_R})_{-\mypi_{R}(\l)},s'_{\mypi_{R}(\l)}\right)$ and $s'_{\mypi_{R}(\l)}$ is a best response of player $\mypi_{R}(\l)$ in $S^{\l-1}_R$.
For any $i=1,\ldots,n$, the last best response performed by player $i$ in $R$ is the $\last_R(i)$-th best response of $R$, leading from state $S^{\last_R(i)-1}$ to state $S^{\last_R(i)}$.
When clear from the context, we will drop the index $R$ from the notation, writing $S^i$, $\mypi$ and $\last(i)$ instead of $S^i_R$, $\mypi_R$ and $\last_R(i)$, respectively.

\begin{definition}[$T$-Minimum Liveness Condition]
Given any $T \geq n$, a best response dynamics satisfies the $T$-Minimum Liveness Condition if it can be decomposed into a sequence of consecutive $T$-coverings.
\end{definition}

In Section \ref{subseqlb} we show that (for the general asymmetric case) under such a condition the quality of the reached state can be very bad even considering $T=O(n)$ (see Corollary \ref{lb_corollary}): It is worth noticing that in the considered congestion game, only $\sqrt[4]{n}$ players perform a lot of best responses ($\sqrt{n}$ best responses) in each covering, while the remaining $n-\sqrt[4]{n}$ players perform only one best response every $T$-covering. The idea here is that there is a sort of unfairness in the dynamics, given by the fact that the players do not have the same chances of performing best responses.

In order to quantify the impact of fairness on best response dynamics, we need an additional parameter $\beta$ and we define the $\beta$-bounded $T$-covering as a $T$-covering in which every player performs at most $\beta$ best responses.


\begin{definition}[$(T,\beta)$-Fairness Condition]
Given any $\beta \geq 1$, a dynamics satisfies the $(T,\beta)$-Fairness Condition if it can be decomposed into a sequence of consecutive $\beta$-bounded $T$-coverings.
\end{definition}

Notice that $\beta$ is a sort of (un)fairness index: If $\beta$ is constant, it means that every player plays at most a constant number of times in each $T$-covering and therefore the dynamics can be considered \emph{fair}.

In order to prove our upper bound results, we will focus our attention on particular congestion games to which any linear congestion game is best-response reducible.  The following definition formally states such a notion of reduction.

\begin{definition}[Best-Response Reduction]
A congestion game ${\cal G}$ is Best-Response reducible to a congestion game ${\cal G}'$ with the same set of players if there exists an injective function $g$ mapping any strategy profile $S$ of ${\cal G}$ to a strategy profile $g(S)$ of ${\cal G}'$ such that
\begin{itemize}
\item[(i)] the social cost of $S$ in ${\cal G}$ is equal to the one of $g(S)$ in ${\cal G}'$ and for any $i=1,\ldots,n$ the cost of player $i$ in $S$ is equal to the one of player $i$ in ${\cal G}'$
\item[(ii)] there exists, in the game ${\cal G}$, a best response of a player, say $i$, in $S$ leading to state $S'$ if and only if there exists, in the game ${\cal G}'$, a best response of  player $i$ in $g(S)$ leading to state $g(S')$.
\end{itemize}
\end{definition}

\section{The Impact of Fairness}\label{sec:asymmetric}

In this section we first (in Subsection \ref{subsequb}) provide an upper bound to the approximation ratio of the states reached after a dynamics satisfying the $(T,\beta)$-Minimum Liveness Condition, starting from an arbitrary state and composed by a number of best responses polynomial in $n$.
Finally (in Subsection \ref{subseqlb}), we provide an almost matching lower bound holding for dynamics satisfying the same conditions.

\subsection{Upper bound} \label{subsequb}
All the results hold for linear congestion games having delay functions $f_e(x)=a_e x +b_e$ with $a_e,b_e \geq 0$ for every $e \in E$.
Since our bounds are given as a function of the number of players, the following proposition allows us to focus on congestion games with identical delay functions $f(x)=x$.

\begin{proposition}\label{prop1}
Any linear congestion game is best-response reducible to a congestion game  having the same set of players and identical delay functions $f(x)=x$.
\end{proposition}

\begin{proof}
Given a congestion game ${\cal G}$ having delay functions
$f_e(x)=a_e x + b_e$ with integer coefficient $a_e, b_e \geq 0$, it is possible
to reduce it to a congestion game ${\cal G}'$, having the same set of players and identical delay functions $f(x)=x$ in the following way. For each resource $e$ in ${\cal G}$, we include in ${\cal G}'$ a set $A_e$ of $a_e$ resources and $n$ sets $B^1_e,\ldots,B^n_e$, each containing $b_e$ resources;
moreover, given any strategy $s_i \in \Sigma_i$ in ${\cal G}$,
$i=1,\ldots,n$, we build a corresponding strategy $s'_i \in \Sigma'_i$
(in ${\cal G}'$) by including in $s'_i$, for each $e \in s_i$, all the
resources in the sets $A_e$ and $B^i_e$.
The function $g$ is such that any strategy profile $(s_1,\ldots,s_n)$ of ${\cal G}$ is mapped to the strategy profile $(s'_1,\ldots,s'_n)$ of ${\cal G}'$.
If $a_e$ and $b_e$ are not integers we can perform a similar reduction by exploiting a simple scaling argument.
\qed
\end{proof}

Since the dynamics satisfies the $(T,\beta)$-Fairness Condition, we can decompose it into $k$ $\beta$-bounded $T$-coverings $R_1,\ldots,R_k$.

Consider a generic $\beta$-bounded $T$-covering $R=\left(S^0,\ldots,S^{\c}\right)$.
In the following we will often consider the {\em immediate} costs (or delays) of players during $R$, that is the cost $c_{\mypi(\l)}(S^\l)$ right after the best response of player $\mypi(\l)$, for $\l=1,\ldots,\c$.

Given an optimal strategy profile $S^*$, since the $\l$-th player $\mypi(\l)$ performing a best response, before doing it, can always select the strategy she would use in $S^*$, her immediate cost can be suitably upper bounded as $\sum_{e\in s_{\mypi(\l)}^*}{\left(n_e(S^{\l-1})+1\right)}$.


By extending and strengthening the technique of \cite{FanelliFM08,FanelliM09}, we are able to prove that the best response dynamics satisfying the $(T,\beta)$-Fairness Condition fast converges to states approximating the social optimum by a factor $O(\beta)$.
It is worth noticing that, by exploiting the technique of \cite{FanelliFM08,FanelliM09}, only a much worse bound of $O(\beta^2)$ could be proved. In order to obtain an $O(\beta)$ bound, we need to develop a different and more involved technique, in which also the functions $\rho$ and $H$, introduced in \cite{FanelliFM08,FanelliM09}, have to be redefined: roughly speaking, they now must take into account only the last move in $R$ of each player, whereas in \cite{FanelliFM08,FanelliM09} they was counting for all the moves in $R$.

We now introduce functions $\rho$ and $H$, defined over the set of all the possible $\beta$-bounded $T$-coverings:
\begin{itemize}
\item Let $\rho(R) =
\sum_{i=1}^{n} {\sum_{e\in
s_i^*}{\left(n_e(S^{\last_R(i)-1})+1\right)}}$;
\item let $H(R) = \sum_{i=1}^{n} {\sum_{e\in s_i^*}{n_e(S^0)}}$.
\end{itemize}

Notice that $\rho(R)$ is an upper bound to the sum over all the players of the cost that she would experience on her optimal strategy $s_{i}^*$ just before her last move in $R$, whereas $H(R)$ represents the sum over all the players of the delay on the moving player's optimal strategy $s_{i}^*$ in the initial state $S^0$ of $R$.
Moreover, since players perform best responses, $\sum_{i=1}^{n} c_{i}(S^{\last_R(i)}) \leq \rho(R)$, i.e. $\rho(R)$ is an upper bound to the sum of the immediate costs over the last moves of every players.

The upper bound proof is structured as follows. Lemma \ref{lem3} relates the social cost of the final state $S^\c$ of a $\beta$-bounded $T$-covering $R$ with $\rho(R)$, by showing that $C(S^{\c}) \leq 2 \rho(R)$. By exploiting Lemmata \ref{lem2} and \ref{lem4}, providing an upper (resp. lower) bound to $H(R)$ in terms of $\rho(\overline{R})$ (resp. $\rho({R})$), Lemma \ref{lem5} proves that $\frac{\rho}{\opt}$ rapidly decreases between two consecutive $\beta$-bounded $T$-coverings $\overline{R}$ and $R$, showing that $\frac{\rho(R)}{\opt} = O\left (\sqrt{\frac{\rho(\overline{R})}{\opt}}\right )$.
In the proof of Theorem $1$, after deriving a trivial upper bound equal to $O(n)$ for $\rho(R_1)$, Lemma \ref{lem5} is applied to all the $k-1$ couples of consecutive $\beta$-bounded $T$-coverings of the considered dynamics satisfying the $(T,\beta)$-Fairness Condition.

The following lemmata show that the social cost at the end of any $\beta$-bounded $T$-covering $R$ is at most $2 \rho(R)$, and that $\frac{\rho}{\opt}$ fast decreases between two consecutive $\beta$-bounded $T$-coverings. They can be proved by adapting some proofs in \cite{FanelliFM08,FanelliM09} so that they still hold with the new definition of $\rho$.

\begin{lemma}\label{lem3}
For any $\beta \geq 1$, given a $\beta$-bounded $T$-covering $R$, $C(S^{\c}) \leq  2 \rho(R)$.
\end{lemma}

\begin{proof}
Since the players perform  best responses, inequality (\ref{riga1}) below holds. Moreover, since in $R$ each player performs at least a best response, inequality (\ref{riga2}) holds because for each resource $e$ used in the state $S^{\c}$ there must exist, for every integer $j$ such that  $1 \leq j \leq n_e(S^{\c})$, a player having, just after her last best response in $R$, a delay on $e$ at least $i$. Therefore, by recalling the definition of $\rho(R)$,
\begin{eqnarray}
\nonumber \rho(R) & = & \sum_{i=1}^{n} {\sum_{e\in s_{i}^*}{\left(n_e(S^{\last(i)-1})+1\right)}} \\
& \geq & \sum_{i=1}^{n} c_{i}(S^{\last(i)}) \label{riga1}\\
& \geq & \sum_{e \in E}\sum_{j=1}^{n_e(S^{\c})} j \label{riga2}\\
\nonumber
& \geq & \frac{1}{2}\sum_{e \in E} n_e^2(S^{\c}) = \frac{1}{2}C(S^{\c}).
\end{eqnarray}
\qed
\end{proof}

\begin{lemma}\label{lem2}
For any $\beta \geq 1$, given a $\beta$-bounded $T$-covering $R$ ending in $S^\c$, $\frac{\sum_{e \in E} {n_e(S^{\c}) n_e(S^*)}}{\opt} \leq  \sqrt{2 \frac{\rho(R)}{\opt}}$.
\end{lemma}
\begin{proof}
By applying the \emph{Cauchy-Schwarz inequality}, we obtain
$$\sum_{e \in E} {n_e(S^{\c}) n_e(S^*)} \leq  \sqrt{\sum_{e \in E}{n^2_e(S^\c)} \sum_{e \in E}{n_e^2(S^*)}}.$$

Moreover, by Lemma \ref{lem3} we have  $ \sum_{e \in E^*}
{n^2_e(S^\c)} = C(S^\c) \leq 2  \rho(\overline{R})$; therefore, since $\sum_{e \in E^*}{n_e^2(S^*)}=\opt$, we obtain
\begin{equation}\label{eqn_UB}
\sum_{e \in E} {n_e(S^{\c}) n_e(S^*)} \leq \sqrt{2\rho(\overline{R})\opt}.
\end{equation}
Dividing by $\opt$ both sides of the above inequality, the claim follows.
\qed
\end{proof}

In Lemma \ref{lem4} we are able to relate $\rho(R)$ and $H(R)$ by much strengthening the technique exploited in \cite{FanelliFM08,FanelliM09}.

\begin{lemma}\label{lem4} For any $\beta \geq 1$, given a $\beta$-bounded $T$-covering $R$, $$\frac{\rho(R)}{\opt}
\leq 2 {\frac{H(R)}{\opt}} +
4 \beta +1.$$
\end{lemma}
\begin{proof}
First of all, notice that if the players performing in $R$ best responses being improvement moves never select strategies used by some player in $S^*$, i.e. if the players select only resources $e$ such that $n_e(S^*)=0$, than, by recalling the definitions of $\rho(R)$ and $H(R)$, $\rho(R) \leq H(R) + \opt$ and the clam would easily follow for any $\beta \geq 1$.

In the following our aim is that of dealing with the generic case in which players moving in $R$ can increase the congestions on resources $e$ such that $n_e(S^*)>0$.

For every resource $e \in E$, we focus on the congestion on such a resource above a ``virtual" congestion frontier $g_e = 2 \beta n_e(S^*)$.

We assume that at the beginning of covering $R$ each resource $e \in E$ has a congestion equal to $\delta_{0,e} = \max\{n_e(S^0),g_e\}$, and we define $\delta_{0,e}$ as the \emph{congestion of level $0$} on resource $e$; moreover, $\Delta_0 = \sum_{e \in E}{\delta_{0,e} \cdot n_e(S^*)}$ is an upper bound to $H(R)$. We refer to $\Delta_0$ as the total congestion of level $0$.

The idea is that the total congestion of level $0$ can induce on the resources a congestion (over the frontier $g_e$) being the total congestion of level $1$, such a congestion a total congestion of level $2$, and so on.
More formally, for any $p \geq 1$ and any $e \in E$, we define $\delta_{p,e}$ as the congestion of level $p$ on resource $e$  above the frontier $g_e$; we say that a congestion $\delta_{p,e}$ of level $p$ on resource $e$ is \emph{induced} by an amount $x_{p,e}$ of congestion of level $p-1$ if some players (say, players  in $N_{p,e}$) moving on $e$ can cause such a congestion of level $p$ on $e$ because they are experimenting a delay on the resources of their optimal strategies due to an amount $x_{p,e}$ of congestion of level $p$. Notice that, for each move of the players in $N_{p,e}$, such an amount $x_{p,e}$ of congestion of level $p-1$ can be used only once, i.e. it cannot be used in order to induce a congestion of level $p$ for other resources in $E \setminus \{e\}$.
In other words, $x_{p,e}$ is the overall congestion of level $p-1$ on the resources in the optimal strategies of players in $N_{p,e}$ used in order to induce the congestion $\delta_{p,e}$ of level $p$ on resource $e$.

For any $p$, the total congestion of level $p$ is defined as $\Delta_p = \sum_{e \in E}{\delta_{p,e}\cdot n_e(S^*)}$. Moreover, for any $p\geq 1$, we have that $\sum_{e \in E} x_{p,e} \leq \beta \Delta_{p-1}$ because each player can move at most $\beta$ times in $R$ and therefore the total congestion of level $p-1$ can be used at most $\beta$ times in order to induce the total congestion of level $p$.

It is worth noticing that $\rho(R) \leq \sum_{p=0}^\infty {\Delta_p}+\opt$, because $\sum_{p=0}^\infty \delta_{p,e}$ is an upper bound on the congestion of resource $e$ during the whole covering $R$:
\begin{eqnarray*}
\rho(R) & =  &  \sum_{i=1}^{n} {\sum_{e\in s_{i}^*}{\left(n_e(S^{\last(i)-1})+1\right)}}\\
& \leq & \sum_{i=1}^{n} \sum_{e\in s_{i}^*}
{\left( \sum_{p=0}^\infty{ \delta_{p,e} } +1 \right)}
=\sum_{e\in E}
\left( n_e(S^*) \left( \sum_{p=0}^\infty \delta_{p,e} +1 \right) \right)\\
& = & \sum_{e\in E}
  \sum_{p=0}^\infty \delta_{p,e}n_e(S^*)   +\sum_{e\in E} n_e(S^*)
= \sum_{p=0}^\infty \Delta_{p}  +\opt
\end{eqnarray*}

In the following, we upper bound $\sum_{p=0}^\infty {\Delta_p}$.
\begin{eqnarray}
\Delta_p & =  & \sum_{e \in E}{\delta_{p,e}\cdot n_e(S^*)}\nonumber\\
& \leq & \sum_{e \in E}{\frac{x_{p,e}}{g_e}\cdot n_e(S^*)}\label{eqn10.1}\\
& \leq & \sum_{e \in E}{\frac{x_{p,e}}{2 \beta n_e(S^*)}\cdot n_e(S^*)}\nonumber\\
& \leq & \frac{\Delta_{p-1}}{2},\label{eqn10.2}
\end{eqnarray}
where inequality \ref{eqn10.1} holds because $\delta_{p,e}$ is the portion of congestion on resource $e$ above the frontier $g_e$ due to some moving players having on the resources of their optimal strategy a delay equal to $x_{p,e}$, and inequality \ref{eqn10.2} holds because each player can move at most $\beta$ times in $R$ and therefore the total congestion of level $p-1$ can be used at most $\beta$ times in order to induce the total congestion of level $p$.

We thus obtain that, for any $p \geq 0$, $\Delta_p  \leq \frac{\Delta_0}{2^p}$ and $\sum_{p=0}^\infty {\Delta_p} \leq 2 \Delta_0$.

Since $\Delta_0 = \sum_{e \in E}{ \max\{n_e(S^0),2\beta n_e(S^*)\} \cdot n_e(S^*)} \leq H(R) + 2\beta \opt $ and $\rho(R) \leq \sum_{p=0}^\infty {\Delta_p}+\opt \leq 2 \Delta_0 + \opt$, we finally obtain the claim.
\qed
\end{proof}

By combining Lemmata \ref{lem2} and \ref{lem4}, it is possible to prove the following lemma showing that $\frac{\rho(R)}{\opt}$ fast decreases between two consecutive coverings.

\begin{lemma}\label{lem5} For any $\beta \geq 1$, given two consecutive $\beta$-bounded $T$-coverings $ \overline{R}$ and $R$, $\frac{\rho(R)}{\opt}
\leq 2  \sqrt{2\frac{\rho(\overline{R})}{\opt}} +
4 \beta +1$.
\end{lemma}
\begin{proof}
Recall that the initial state $S^0$ of walk $R$ coincides with the final state of walk $\overline{R}$, and that $H(R)= \sum_{i=1}^{\c} {\sum_{e\in s_{\mypi(i)}^*}{\left(n_e(S^0)\right)}}$.
By inverting the summation in the definition of $H(R)$, we obtain that $H(R)=\sum_{e \in E} n_e(S^0)n_e(S^*)$. Therefore, by Lemma \ref{lem2}, $\frac{H(R)}{\opt} \leq  \sqrt{2 \frac{\rho(\overline{R})}{\opt}}$. Hence, by combining such an inequality with Lemma \ref{lem4}, the claim follows.
\qed
\end{proof}

By applying Lemma \ref{lem5} to all the couples of consecutive $\beta$-bounded $T$-coverings, we are now able to prove the following theorem.

\begin{theorem} \label{ub}
For any linear congestion game $\cal G$, any best response dynamics satisfying the $(T,\beta)$-Fairness Condition converges from any initial state to a state $S$ such that $\frac{C(S)}{\opt} = O(\beta)$ in at most $T \lceil \log \log n \rceil$ best responses.
\end{theorem}
\begin{proof}
Given a bast response dynamics satisfying the $(T,\beta)$-Fairness Condition, let $R_1,\ldots,R_k$ be the $k$ $\beta$-bounded $T$-coverings in which it can be decomposed.
By applying Lemma \ref{lem5} to all the pairs of consecutive  $\beta$-bounded $T$-coverings $R_j$ and $R_{j+1}$, for any $j=1,\ldots,k-1$ we obtain that for any $\alpha >1$
$$\frac{\rho(R_{j+1})}{\opt}
\leq 2  \sqrt{2\frac{\rho(R_{j})}{\opt}} +
4 \beta +1.$$

By combining all the above inequalities for $j=1,\ldots,k-1$ and by performing some basic algebraic manipulations, we obtain that
$\frac{\rho(R_k)}{\opt} =O\left(\sqrt[2^{k-1}]{\frac{\rho(R_1)}{\opt}}+\beta\right)$.
Thus, by Lemma \ref{lem3}, the cost of the final state $S$ of walk $R_k$ is $$\frac{C(S)}{\opt} = O\left(\sqrt[2^{k-1}]{\frac{\rho(R_1)}{\opt}} + \beta\right).$$

By the definition of $\rho(R)$, since $\sum_{e\in E}{n_e(S^*)} \leq \sum_{e\in E}{n_e^2(S^*)} = \opt $, for any possible $\beta$-bounded $T$-covering $R$ it holds that
\begin{eqnarray*}
\rho(R) & =  &  \sum_{i=1}^{n} {\sum_{e\in s_{i}^*}{\left(n_e(S^{\last(i)-1})+1\right)}} \leq
\sum_{i=1}^{n} {\sum_{e\in s_{i}^*}{(n+1)}}\\
& = & (n+1) \sum_{i=1}^{n} |s_{i}^*| \leq (n+1) \sum_{e\in E}{n_e(S^*)} \leq (n+1) \opt.\\
\end{eqnarray*}

Therefore, $\frac{\rho(R_1)}{\opt} \leq n+1$ and we obtain
$$\frac{C(S)}{\opt} = O\left(\sqrt[2^{k-1}]{n} + \beta\right).$$

It is worth noticing that $k = \log \log n$ $\beta$-bounded $T$-coverings are sufficient in order to obtain $\frac{C(S)}{\opt} = O(\beta).$
Since every $\beta$-bounded $T$-covering, by its definition, contains at most $T$ best responses, the claim follows.
\qed
\end{proof}

\subsection{Lower bound} \label{subseqlb}

\begin{theorem} \label{lb}
For any $\epsilon >0 $, there exist  a linear congestion game $\cal G$ and an initial state $S^0$ such that, for any $\beta = O(n^{-\frac{1}{\log_2 \epsilon}})$, there exists a best response dynamics starting from $S^0$ and satisfying the $(T,\beta)$-Fairness Condition such that for a number of best responses exponential in $n$ the cost of the reached states is always $\Omega(\beta^{1-\epsilon} \cdot \opt)$.
\end{theorem}
%
%
%
%
%
%
%
%
%
%
%
%
%
%
%
%
%
%
%

%
%
\begin{proof}
We construct a congestion game $\cal G$ for which there is an initial state $S_0$ and an exponentially long sequence of best responses starting from this state. In every round, each player moves at most $\beta$ times. For technical reasons, we assume $\beta$ is even and $\beta>10$.

We present the construction in two stages. In a first step we construct a game  $\cal G'$ that has the desired properties only if we assume players change their strategies even if they have exactly the same delay.
In a second step, extend  $\cal G'$ towards a game  $\cal G$ in which those strategy changes are improving best responses.

We define $f_i =(\frac{\beta}{\log \beta})^{1-\frac{1}{2^{i+1}}}$.
The players and the resources are divided into $L+1$ levels, $0,\ldots,L$. Each level $i$ is divided into $\beta^i$ blocks $(i,1),\ldots,(i,\beta^i)$.
Each block $(i,j)$ consists of $3 f_i$ players, $\frac{m}{\beta^i}$ {\em main} resources, and $2 f_i \log \beta$ {\em address} resources.
We group the players into $f_i$ triplets and denote them by $P_{i,j,k}$, $Q_{i,j,k}$, and $R_{i,j,k}$ with $1 \le k \le f_i$.
The general idea is the following. In the optimal solution all players of level $i$ only use the main resources of level $i$. The resources are evenly distributed among the triplets. However, in the sequence that we construct, the $P$ player of a block of level $i$ always allocate main resources of a block in level $i+1$ such that these players cannot choose the main resources of level $i+1$. We say the $P$ players of a block $(i,j)$ {\em block} the
players of a block $(i+1, j')$. In each round, each block blocks $\beta$ blocks of the next level. The purpose of the $Q$ and $R$ players and the address resources is to make sure that all $P$ players block the same block of the next level.

We now describe the strategies and resources in detail. Each player has $\beta+1$ strategies that we denote by $s_0(\cdot),\ldots,s_\beta(\cdot)$.
Only the $s_0$ strategies of players of a block $(i,j)$ contain main resources from level $i$. Moreover, the main resources of block $(i,j)$ are equally distributed among the triplets. The delay function of the main resources is $f(x) = x$.\\

$s_0(P_{i,j,k}) = s_0(Q_{i,j,k}) = s_0(R_{i,j,k}) = \{p_{i,j,\frac{m(k-1)}{\beta^i f_i}+1},\ldots, p_{i,j,\frac{m k}{\beta^i f_i}} \}$\\

The $Q$ and $R$ players use the address resources to binary encode $\alpha$ when playing $s_\alpha$. The $Q$ players are for the even numbers the $R$ players for the odd numbers. Thus, for every $1 \le \alpha \le \beta$ the $Q$ and $R$ players have a strategy.\\

$s_\alpha(Q_{i,j,k}) = \{q^0_{i,j,k,l} \mid \text{for all } 1 \le l \le \log(\beta -1) \text{ with bit } l \text{ of } \alpha \text{ is } 0 \} \cup 
                       \{q^1_{i,j,k,l} \mid \text{for all } 1 \le l \le \log(\beta -1) \text{ with bit } l \text{ of } \alpha \text{ is } 1 \}$\\

$s_\alpha(R_{i,j,k}) = \{r^0_{i,j,k,l} \mid \text{for all } 1 \le l \le \log(\beta -1) \text{ with bit } l \text{ of } \alpha \text{ is } 0 \} \cup 
                       \{r^1_{i,j,k,l} \mid \text{for all } 1 \le l \le \log(\beta -1) \text{ with bit } l \text{ of } \alpha \text{ is } 1 \}$\\
                       
The delay function for each of the $r$ and $q$ resources is $f(x) = \frac{m f_i}{\beta^{i+1}} x$.

The $P$ players allocate main resources of the next level $(i+1)$ in their strategies $s_\alpha$.
More precisely, the players of block $(i,j)$ allocate all main resources that are part of the $s_0$ strategy of block $(i+1,(j-1)\beta + \alpha)$. 
Additionally, the strategy contain the address resources that we used above to binary encode a number. Note that the use of the address resources is inverted compared to the $Q$ and $R$ players. \\

For all even $\alpha \in \{2,\ldots,\beta\}$:\\
$s_\alpha(P_{i,j,k}) = \{t_{i,j,k,\alpha},  p_{i+1,(j-1)\beta + \alpha,1},\ldots, p_{i+1,(j-1)\beta + \alpha,\frac{m}{\beta^{i+1}}} \}
 \cup \{q^1_{i,j,k,l} \mid \text{for all } 1 \le l \le \log(\beta -1) \text{ with bit } l \text{ of } \alpha \text{ is } 0 \} \cup 
      \{q^0_{i,j,k,l} \mid \text{for all } 1 \le l \le \log(\beta -1) \text{ with bit } l \text{ of } \alpha \text{ is } 1 \}$\\
         
For all odd $\alpha \in \{1,\ldots,\beta-1\}$:\\
 $s_\alpha(P_{i,j,k}) =  \{t_{i,j,k,\alpha}, p_{i+1,(j-1)\beta + \alpha,1},\ldots, p_{i+1,(j-1)\beta + \alpha,\frac{m}{\beta^{i+1}}} \}
    \cup\{r^1_{i,j,k,l} \mid \text{for all } 1 \le l \le \log(\beta -1) \text{ with bit } l \text{ of } \alpha \text{ is } 0 \} \cup 
          \{r^0_{i,j,k,l} \mid \text{for all } 1 \le l \le \log(\beta -1) \text{ with bit } l \text{ of } \alpha \text{ is } 1 \}$\\

The resource $t_{i,j,k,\alpha}$ has constant delay of $\frac{m (f_i - k)}{\beta^{i+1}}$ if $\alpha$ is odd and delay of $\frac{m (k-1)}{\beta^{i+1}}$ otherwise.

Before we describe the sequence of best responses let us observe some useful properties of the construction.
Since there are $f_i$ player in every block of level $i$, we can conclude the following.
\begin{claim}
If all $f_i$ $P$ players of a block $(i,j)$ play strategy $s_\alpha$ for some $\alpha$, then each resource of block $(i+1,(j-1)\beta + \alpha)$ has delay of at least $f_i$.
\end{claim}

We now observe that there a useful sequences of strategy changes in which the delay of the moving player stays the same.
\begin{proposition}\label{lem}
For an even $\alpha$, consider a strategy profile in which the players $P_{i,j,k}$ and $Q_{i,j,k}$ of a group $(i,j)$ play strategy $s_\alpha$, and the players $R_{i,j,k}$ play $s_{\alpha + 1}$. If the players  $P_{i,j,k}$ change  to $s_{\alpha + 1}$ in an order increasing in $k$, then each player has delay of $ \frac{m f_i \log \beta}{\beta^{i+1}}$ before and after his strategy change. Additionally, when a player moves, his strategies $s_{\alpha'}$ with $\alpha' \not\in \{\alpha, \alpha+1\}$ have delay of more $\frac{m f_i \log \beta}{\beta^{i+1}}$.
\end{proposition}

\begin{proof}
Consider a player $P_{i,j,k}$ before he moves to $s_{\alpha + 1}$. Then $k-1$ many $P$ players have already changed to $S_{\alpha +1}$. Thus, the $p$ resources of block $(i+1,(j-1)\beta + \alpha)$ are used by $f_i - (k -1)$ other players and, therefore, induce total delay of $\frac{m}{\beta^{i+1}} (f_i - k +1)$. Adding  $\frac{m (k-1)}{\beta^{i+1}}$ for the $t$ resource and $(\log (\beta-1))  \frac{m f_i}{\beta^{i+1}}$ for the address resources gives the desired result for the strategy $s_\alpha$.

If he changes to $S_{\alpha +1}$, there will be $k$ players on the the $p$ resources of block $(i+1,(j-1)\beta + \alpha +1)$ inducing total delay of 
 $\frac{m}{\beta^{i+1}} k$. Again, with the $t$ resource and the address resources that conveniently adds up to  $ \frac{\log \beta m f_i}{\beta^{i+1}}$.
 
For the strategies $s_{\alpha'}$ with $\alpha' \not\in \{\alpha, \alpha+1\}$, first consider the case of $\alpha'$ being even.
Observe that there is a least one bit in which $\alpha'$ differs from the even number in $\{\alpha, \alpha+1\}$. Therefore, at least one of the $q$ resources in strategy $s_{\alpha'}$  of the $P$ player is also used by the $Q$ player of the same triplet and, thus, induces delay of $2 \frac{m f_i}{\beta^{i+1}}$. The delay of the address resources adds up to at least $\frac{m f_i \log \beta}{\beta^{i+1}}$ and taking into account the delay of the $t$ resource the claim follows. The case of an odd $\alpha'$ is analogous.
\qed
\end{proof}

\begin{claim}
The delay of a blocked strategy $s_0$ is greater than a state in the sequence described in Proposition~\ref{lem}, i.e., $\frac{m}{\beta^i f_i}f_{i-1} > \frac{m f_i \log \beta}{\beta^{i+1}}$.  
\end{claim}

We are now ready to describe the sequence of strategy changes.
We start in a  state $S_0$ in which every player plays his strategy $s_\beta$. The sequence is recursively defined as described by Algorithm~\ref{Seq1} which is started by run$(1,1)$.
Note that every player moves at most $\beta$ time is this sequence. We call this sequence a {\em run} of $\cal G'$.
In a nutshell: 
Whenever a block is blocked, the $P$-players iteratively block the $s_0$ strategies of certain $\beta$ blocks of the next level.
The $Q$ and $R$ players make sure, that the $P$ players block the  $s_0$ strategies of the same block, i.e, all $P$ players
 change to their strategy $s_\alpha$ for some $\alpha$.\\

\begin{algorithm}
\caption{Recursive procedure run$(i,j)$ } \label{Seq1}
\begin{algorithmic} [1]
\FOR{$\alpha=1$ \TO $\beta$} 
  \IF	{$\alpha$ is even}
	  \FORALL{$1 \le k \le f_i$}
	    \STATE $Q_{i,j,k}$ changes to $s_\alpha(Q_{i,j,k})$
	  \ENDFOR
	  \FOR{$k = 1$ \TO $f_i$}
	      \STATE $P_{i,j,k}$ changes to $s_\alpha(P_{i,j,k})$
	  \ENDFOR 	
	\ELSE
	  \FORALL{$1 \le k \le f_i$}
	    \STATE $R_{i,j,k}$ changes to $s_\alpha(R_{i,j,k})$
	  \ENDFOR
	  \FOR{$k = f_i$ \TO $1$}
	      \STATE $P_{i,j,k}$ changes to $s_\alpha(P_{i,j,k})$
	  \ENDFOR 
	\ENDIF
  \STATE run$(i+1,(j-1)\beta + \alpha)$
\ENDFOR
\end{algorithmic} 
\end{algorithm}

In the optimal solution, all players play their $s_0$ strategy which yields social cost of $3 m (L+1)$. During our sequence the cost of level $L$ is $m f_{L-1} = (\frac{\beta}{\log \beta})^{1-\frac{1}{2^{L}}}$.
Therefore, given any $\epsilon>0$, it suffices to consider an $\epsilon' < \epsilon$ and to choose $L = -\log{\epsilon'}$  so that the approximation ratio of all the states of the considered dynamics is always at least $\frac{\beta^{1-\epsilon'}}{3 (1- \log{\epsilon'})\log^{1-\epsilon'} \beta} = \Omega(\beta^{1-\epsilon})$. 

Note that the in the above sequence, the delay of a player that changes his strategy does not change. In order to complete the proof, we now turn the game into one in which these strategy changes are best responses. Let us give an overview of the modifications first.
We combine above game  $\cal G'$  with a game  $G_n$  that has the property that there is an initial state and every sequence of best responses starting from that state is exponentially long. We essentially use the game that was presented in \cite{SkopalikV08} and modify its delay functions to make them linear. The best response sequence of $G_n$ resembles the recursive run of $n$ programs in which program $i$ is executed $2^{n-i}$ times. We use the constants $M > 20\delta^n$ and $\delta \gg |N|(\beta+1)$ as scaling factors, where $N$ denotes the set of players of $\cal G'$. Note that they are independent of $m$ which allows us to arbitrarily scale the cost of $G'$.

To each strategy of  $\cal G'$, we add a unique {\em trigger} resource with the delay function $f(x) = x$. Note that this does not change the preferences of the players. We extend the game with additional {\em trigger} players. 
Whenever a player of  $\cal G'$  is supposed to change to a strategy $s$, we make sure that the trigger resources in all his other strategies are used by trigger players. Therefore, his strategy change decreases his delay by $1$.

Let $r$ denote the number of strategy changes of a run of $\cal G'$,
We extend the game with $r$ trigger players for each strategy of each player of $\cal G'$. Each trigger player has two strategies, trigger and wait. If he plays trigger, he allocates the trigger resource. His wait strategy is contains a resource that is also contained in the strategies of a player of game  $G_n$. 
We  have one {\em reset} player for each trigger player. He has a strategy wait which is also connected to a player of game $G_n$ and a strategy reset. If he plays reset, the corresponding trigger player has an incentive to change back to wait.

We are now ready to describe $\cal G$ in detail.
For each $i \in N$ and $0 \le \alpha \le \beta$, we add a resource $v_{i,\alpha}$ to strategy $s_\alpha(i)$. It has the delay function $f(x) = x$.
 For each $i \in N$ and $0 \le \alpha \le \beta$, we add a player $T_{i,\alpha}$. His strategies are described in Figure~\ref{figure:Trigger}.
We add $n$ players  Main$_i$ as described in Figure~\ref{figure:Main} and $8n$ players Block$^j_i$ as described in Figure~\ref{figure:Block}.
Note that resources of the wait strategies of the trigger and reset players are part of strategy (9) of player  Main$_1$.

We now describe the best response sequence of $\cal G$. Algorithm~\ref{seq2} describes a recursive run of the players of $G'$, the trigger players, and reset players. Note that in such a run, every player moves at most $\beta$ times. We now combine such runs with the best response dynamics of $G_n$.
Note that the best response sequence of  $G_n$ is essentially unique as described in \cite{SkopalikV08}. We divide this exponentially long sequence into rounds. We define the beginning of a round as the state before player Main$_1$ moves from strategy (8) to (9). 
This divides the best response sequence into exponentially many rounds in which every player of $G_n$ moves at most $9$ times. 
At the beginning of a round, we interrupt this sequences and let the remaining players move. 
We first let all trigger and reset players play their best response. This ensures that they change to wait. We then let Main$_1$ move from strategy (8) to (9). We then start one run as described in Algorithm~\ref{seq2}.  Finally, we let the players of $G_n$ play their best response in decreasing order of their indices with player Block$_i^8$ as the last player. He is also the only one that changes his strategy. With his move, we continue the best response sequence of $G_n$ for the remainder of the round.

This yields a exponentially long sequence of best responses that satisfies the $(T,\beta)$-Fairness Condition.

\begin{figure}[th]
\begin{center}
\begin{tabular}[ht]{|l|l|l|}
  \hline
  Strategies of $T_{i,\alpha,d}$  & Resources & Delay function $f(x)$ \\
  \hline 
\hline

trigger   & $v_{i,\alpha}$                 & $x$ \\
          & $u_{i,\alpha,d}$                 & $2x$ \\

\hline
wait       &$w_{i,\alpha,d}$                &  $2 x$ \\
\hline
\end{tabular}
\end{center}
\caption{\label{figure:Trigger}
Definition of the strategies of the trigger player $T_{i,\alpha,d}$}
\end{figure}

\begin{figure}[th]
\begin{center}
\begin{tabular}[ht]{|l|l|l|}
  \hline
  Strategies of $U_{i,\alpha,d}$  & Resources & Delay function $f(x)$ \\
  \hline 
\hline
reset     & $u$                               & $3$ \\
          & $u_{i,\alpha,d}$                 & $2x$ \\

\hline
wait       &$x_{i,\alpha,d}$                &  $4 x$ \\
\hline
\end{tabular}
\end{center}
\caption{\label{figure:Reset}
Definition of the strategies of the reset player $U_{i,\alpha,d}$}
\end{figure}

\begin{figure}[th]
\begin{center}
\begin{tabular}[ht]{|l|l|l|}
  \hline
  Strategies of Block$^j_i$  & Resources & Delay function $f(x)$ \\
  \hline 
\hline

(1)       & $t^j_i$                 & $x$ \\
          & $b^j_i$                 & $M x$ \\

\hline
(2)       & $c^1_{i}$               &  $M x$ \\
\hline
\end{tabular}
\end{center}
\caption{\label{figure:Block}
Definition of the strategies of the players Block$^j_i$}
\end{figure}

\begin{figure}[th]
\begin{center}
\begin{tabular}[ht]{|l|l|l|}
  \hline
  Strategy   & Resources & Delay function $f(x)$ \\
  \hline 
\hline

(1)       & $e^1_i$                  &   $10\delta^i x  $\\
          & $e^0_i$                  &   $10M$\\
                    
\hline
(2)       & $e^2_i$                  &  $ 2M +19 \delta^i$\\
          & $c^1_{i-1},\ldots,c^8_{i-1}$&  $M x$ \\
          & $t^1_i$                  & $x$ \\

\hline
(3)       & $e^3_i$            & $9M + 18\delta^i$\\
          & $e^1_{i-1}$        & $10\delta^{i-1}x$ \\

          & $t^2_i$            &  $x$ \\
          & $b^1_i$            &  $Mx$ \\
  
\hline
(4)       & $e^4_i$             & $8M + 17 \delta^i$\\
          & $b^8_{i-1}$         & $M x$\\

          & $t^3_i$            &  $x$ \\
          & $b^2_i$            &  $M x$ \\
\hline

(5)       & $e^5_i$             & $9M + 16 \delta^i$\\

          & $t^4_i$            &  $x$ \\
          & $b^3_i$            &  $M x$ \\
\hline
(6)       & $e^6_i$             & $ M+15 \delta^i$\\
          & $c^1_{i-1},\ldots,c^8_{i-1}$&   $Mx$\\

          & $t^5_i$            &  $x$ \\
          & $b^4_i$            &  $Mx$ \\
\hline
(7)       & $e^7_i$             &  $9M + 14 \delta^i$\\
          & $e^1_{i-1}$         &  $10 \delta^{i-1}$\\
  
          & $t^6_i$            &  $x$ \\
          & $b^5_i$            &  $M$ \\
\hline
(8)       & $e^8_i$            &  $8M + 13 \delta^i$\\
          & $b^8_{i-1}$        &  $M x$  \\
          & $t^7_i$            &  $x$ \\
          & $b^6_i$            &  $M x$ \\
\hline
(9)       & $e^9 $             &  $9M + 12\delta^i$\\

          & $t^8_i$            &  $x$ \\
          & $b^7_i$            &  $Mx$ \\
  \hline
\end{tabular}
\end{center}
\caption{\label{figure:Main} Definition of the strategies of the
  players Main$_i$. The delay of resource $e_n^1$ is constantly
  $9 \delta^n$. Strategy (9) of player Main$_1$ additionally contains resources $w_{i,\alpha,d}$ and $x_{i,\alpha,d}$ for all $i \in N$ and $0 \le \alpha \le \beta$ and $1 \le d \le r$.}
\end{figure}

\begin{algorithm}
\caption{Recursive procedure run$_2(i,j)$ } \label{seq2}
\begin{algorithmic} [1]
\FOR{$\alpha=1$ \TO $\beta$} 
  \IF	{$\alpha$ is even}
	  \FORALL{$1 \le k \le f_i$}
	  	        \FORALL{$\alpha' \ne \alpha$}
	            \STATE $T_{(Q_{i,j,k}),\alpha',d}$ change to trigger 
	        \ENDFOR
	    \STATE $Q_{i,j,k}$ changes to $s_\alpha(Q_{i,j,k})$
	        \FORALL{$\alpha' \ne \alpha$}
	            \STATE $U_{(Q_{i,j,k}),\alpha',d}$ change to reset 
	            \STATE $T_{(Q_{i,j,k}),\alpha',d}$ change to wait 
	        \ENDFOR
	  	    \STATE $d=d+1$
	 	  \ENDFOR
	  \FOR{$k = 1$ \TO $f_i$}
	        \FORALL{$\alpha' \ne \alpha$}
	            \STATE $T_{(P_{i,j,k}),\alpha',d}$ change to trigger 
	        \ENDFOR  
	      \STATE $P_{i,j,k}$ changes to $s_\alpha(P_{i,j,k})$
	        \FORALL{$\alpha' \ne \alpha$}
	            \STATE $U_{(P_{i,j,k}),\alpha',d}$ change to reset 
	            \STATE $T_{(P_{i,j,k}),\alpha',d}$ change to wait 
	        \ENDFOR
		    \STATE $d=d+1$
	  \ENDFOR 
	\ELSE
	  \FORALL{$1 \le k \le f_i$}
	        \FORALL{$\alpha' \ne \alpha$}
	            \STATE $T_{(R_{i,j,k}),\alpha',d}$ change to trigger 
	        \ENDFOR		  
	    \STATE $R_{i,j,k}$ changes to $s_\alpha(R_{i,j,k})$
	    \STATE $d=d+1$
	        \FORALL{$\alpha' \ne \alpha$}
	            \STATE $U_{(R_{i,j,k}),\alpha',d}$ change to reset 
	            \STATE $T_{(R_{i,j,k}),\alpha',d}$ change to wait 
	        \ENDFOR		    
	  \ENDFOR
	  \FOR{$k = f_i$ \TO $1$}
	        \FORALL{$\alpha' \ne \alpha$}
	            \STATE $T_{(P_{i,j,k}),\alpha',d}$ change to trigger 
	        \ENDFOR
	      \STATE $P_{i,j,k}$ changes to $s_\alpha(P_{i,j,k})$
	      \STATE $d=d+1$
	        \FORALL{$\alpha' \ne \alpha$}
	            \STATE $U_{(P_{i,j,k}),\alpha',d}$ change to reset 
	            \STATE $T_{(P_{i,j,k}),\alpha',d}$ change to wait 
	        \ENDFOR	      
	  \ENDFOR 
	\ENDIF
  \STATE run$(i+1,(j-1)\beta + \alpha)$
\ENDFOR
\end{algorithmic} 
\end{algorithm}

\qed
\end{proof}


By choosing $\beta=\sqrt{n}$ and considering a simplified version of the proof giving the above lower bound, it is possible to prove the following corollary. In particular, it shows that even in the case of best response dynamics verifying an $O(n)$-Minimum Liveness Condition, the speed of convergence to efficient states is very slow; such a fact implies that the $T$-Minimum Liveness condition cannot precisely characterize the speed of convergence to efficient states because it does not capture the notion of fairness in best response dynamics.

\begin{corollary} \label{lb_corollary}
There exist  a linear congestion game $\cal G$, an initial state $S^0$ and a best response dynamics starting from $S^0$ and satisfying the $O(n)$-Minimum Liveness Condition such that for a number of best responses exponential in $n$ the cost of the reached states is always $\Omega\left(\frac{\sqrt[4]{n}}{\log n} \cdot \opt\right)$.
\end{corollary}
\begin{proof}
The linear congestion game proving the corollary can be easily obtained by simplifying the construction exploited in the proof of Theorem \ref{lb} such that $\beta= \sqrt{n}$ and the number of levels is $2$ (i.e., $L=1$). In such a way, $\Theta(\frac{\sqrt[4]{n}}{\log n})$ players belong to level $0$ and the remaining players to level $1$.

In the optimal solution, all players play their $s_0$ strategy which yields social cost of $6 m$. During our sequence the cost of level $1$ is $m f_{0} = \left(\frac{\beta}{\log \beta}\right)^{\frac{1}{2}}=\Omega\left( \frac{\sqrt[4]{n}}{\log n}\right)$.

\qed
\end{proof}


\section{Symmetric Congestion Games}\label{sec:symmetric}

In this Section we show that in the symmetric case the unfairness in best response dynamics does not affect the speed of convergence to efficient states. In particular, we are able to show that, for any $\beta$, after $T \lceil \log \log n \rceil$ best responses an efficient state is always reached. To this aim, in the following we consider best response dynamics satisfying only the $T$-Minimum Liveness Condition, i.e. decomposable into $k$ $T$-coverings $R_1,\ldots,R_k$.

All the results hold for linear congestion games having delay functions $f_e(x)=a_e x +b_e$ with $a_e,b_e \geq 0$ for every $e \in E$.
Analogously to the asymmetric case, since our bounds are given as a function of the number of players, the following proposition allows us to focus on congestion games with identical delay functions $f(x)=x$.

\begin{proposition}\label{prop2}
Any symmetric linear congestion game is best-response reducible to a symmetric congestion game  having the same set of players and identical delay functions $f(x)=x$.
\end{proposition}
\begin{proof}
Given a symmetric congestion game ${\cal G}$ having delay functions $f_e(x)=a_e x + b_e$ with integer coefficient $a_e, b_e \geq 0$, it is possible to perform a best-response reduction to a symmetric congestion game ${\cal G}'$, having the same set of players and identical delay functions $f(x)=x$ in the following way. For each resource $e$ in ${\cal G}$, we include in ${\cal G}'$ a set $A_e$ of $a_e$ resources and $n$ sets $B^1_e,\ldots,B^n_e$, each containing $b_e$ resources; moreover, given  any strategy set $s_i \in \Sigma_i$ in ${\cal G}$, $i=1,\ldots,n$, we build $n$ corresponding strategies $s'_{i,1}, s'_{i,2},\ldots, s'_{i,n} \in \Sigma'_i$ (in ${\cal G}'$) by including in $s'_{i,j}$, for each $e \in s_i$, all the resources in the sets $A_e$ and $B^j_e$. If $a_e$ and $b_e$ are not integers we can perform a similar reduction by exploiting a simple scaling argument.

The function $g$ is such that any strategy profile $S=(s_1,\ldots,s_n)$ of ${\cal G}$ is mapped to the strategy profile $S'=(s'_{1,1},s'_{2,2}\ldots,s'_{n,n})$ of ${\cal G}'$. In such a way, we make sure that, if player $j$ is using resource $e$ in $S$, in $S'$ only player $j$ is using the strategy associated with the set $B^j_{e}$. In fact, the reduction guarantees that at every best response of the dynamics, for each resource $e$, at most one player uses the set of resources $B^k_e$, for every $k$.

If $a_e$ and $b_e$ are not integers we can perform a similar reduction by exploiting a simple scaling argument.
\qed
\end{proof}

Consider a generic $T$-covering $R=\left(S^0,\ldots,S^{\c}\right)$.
Given an optimal strategy profile $S^*$, since the $\l$-th player $\mypi(\l)$ performing a best response, before doing it, can always select any strategy $s^*_i$, for $i=1\ldots n$, of $S^*$, her immediate cost $c_{\mypi(\l)}(S^\l)$ can be upper bounded as  $\frac{1}{n}\sum_{i=1}^n\sum_{e\in s^*_i}f_e(n_e(S^{\l-1}) + 1) = \frac{1}{n}\sum_{e\in E} n_e(S^*) f_e(n_e(S^{\l-1}) + 1)$. In order to prove our upper bound result, we introduce the following function:
\begin{itemize}
\item $\Gamma(R)=\frac{1}{n}\sum_{i=1}^n\sum_{e\in E} n_e(S^*) f_e(n_e(S^{\last(i)-1}) + 1)$.
\end{itemize}

Notice that $\Gamma(R)$ is an upper bound to the sum of the immediate cost over the last moves of every players, i.e., $\Gamma(R) \geq \sum_{i=1}^n c_i(S^{\last_R(i)})$. Therefore, by exploiting the same arguments used in the proof of Lemma \ref{lem3}, it is possible to prove the following lemma relating the social cost $C(S^\c)$ at the end of $R$ with $\Gamma(R)$.

\begin{lemma}\label{lemma1}
Given any $T$-covering $R$, $C(S^\c) \leq 2 \Gamma(R)$.
\end{lemma}

Moreover, given any $T$-covering $R$, we can relate the social cost $C(S^\c)$ of the final state of $R$ with the cost $C(S^0)$ of its initial state.

\begin{lemma}\label{lemma2}
Given any $T$-covering $R$, $\frac{C(S^\c)}{\opt} \leq (2+2\sqrt{2}) \sqrt{\frac{C(S^0)}{\opt}}$.
\end{lemma}
\begin{proof}
\begin{eqnarray}
\frac{C(S^\c)}{\opt} & \leq & \frac{2\Gamma(R)}{\opt} \label{eqn:1.1}\\
& = &  \frac{2}{n \opt }\sum_{i=1}^n\sum_{e\in {E}} n_e(S^*) (n_e(S^{\last(i)-1}) +1)\nonumber\\
& = &  \frac{2}{n \opt } \left ( \sum_{i=1}^n\sum_{e\in {E}} n_e(S^*) n_e(S^{\last(i)-1}) + \sum_{i=1}^n\sum_{e\in {E}} n_e(S^*)\right ) \nonumber\\
& \leq & 2 + \frac{2}{n \opt} \sum_{i=1}^n\sum_{e\in {E}} n_e(S^*) n_e(S^{\last(i)-1})\nonumber\\
& \leq & 2 +  \frac{2}{n \opt}\sum_{i=1}^n \sqrt{\sum_{e\in {E}} n^2_e(S^*)} \sqrt{ \sum_{e\in {E}} n^2_e(S^{\last(i)-1})}\label{eqn:1.3}\\
& = & 2+ \frac{2}{n \opt}\sum_{i=1}^n \sqrt{\opt} \sqrt{ C(S^{\last(i)-1})}
\nonumber\\
& \leq & 2 + \frac{2 }{n \sqrt{\opt}} \sum_{i=1}^n \sqrt{ 2 \Phi(S^{\last(i)-1})}\label{eqn:1.4}\\
& \leq &  2 + \frac{2  }{n \sqrt{\opt}} \sum_{i=1}^n \sqrt{2 \Phi(S^0)}\label{eqn:1.5}\\
& \leq & 2+2\sqrt{2} \sqrt{\frac{ C(S^0)}{\opt}}\label{eqn:1.6}\\
& \leq & (2+2\sqrt{2}) \sqrt{\frac{ C(S^0)}{\opt}},\nonumber
\end{eqnarray}
where inequality (\ref{eqn:1.1}) follows from Lemma \ref{lemma1},
inequality (\ref{eqn:1.3}) is due to the application of the Cauchy-Schwarz inequality, inequality (\ref{eqn:1.4}) holds because $C(S^{\last(i)-1}) \leq 2 \Phi(S^{\last(i)-1})$, inequality (\ref{eqn:1.5}) holds because the potential function can only decrease at each best response and inequality (\ref{eqn:1.6}) holds because $\Phi(S^0) \leq C(S^0)$.
\qed
\end{proof}

By applying Lemma \ref{lemma2} to all the couples of consecutive  $T$-coverings, we are now able to prove the following theorem.

\begin{theorem}\label{thm1}
Given a linear symmetric congestion game, any best response dynamics satisfying the $T$-Minimum Liveness Condition converges from any initial state to a state $S$ such that $\frac{C(S)}{\opt}=O(1)$ in at most  $T \lceil \ln \ln n \rceil$ best responses.
\end{theorem}
\begin{proof}
We decompose the dynamics into $k$ consecutive $T$-coverings $R_1,\ldots,R_k$.
Recall that, for any $j=1,\ldots,k$, $S^0_j$ and $S^\c_j$ are the initial and the final state of covering $R_j$, respectively. Notice that for any $j=1,\ldots,k-1$, $S^\c_j=S^0_{j+1}$ and $S=S^\c_k$.

By applying Lemma \ref{lemma2} $k-1$ times, we obtain
$$\frac{C(S^\c_k)}{\opt} \leq (2+2\sqrt{2})^2 \left(\frac{ C(S^0_1)}{\opt}\right)^{\frac{1}{2^{k-1}}}.$$

Since, at it is easy to check, $\frac{ C(S^0_1)}{\opt} \leq n$, the claim follows by choosing $k=\lceil \ln \ln n\rceil$.
\qed
\end{proof}



\section{Conclusion and Extensions}\label{sec:conclusion}
In this work we have completely characterized how, in linear congestion games, the frequency with which each player participates in the game dynamics affects the possibility of reaching states with an approximation ratio within a constant factor from the price of anarchy, within a polynomially bounded number of best responses. We have shown that, while in the asymmetric setting the fairness among players is a necessary and sufficient condition for guaranteeing a fast convergence to efficient states, in the symmetric one the game always converges to an efficient state after a polynomial number of best responses, regardless of the frequency each player moves with.

It is worth to note that our techniques provide a much faster convergence to efficient states with respect to the previous result in the literature. In particular, in the symmetric setting, Theorem \ref{thm1} shows that best response dynamics leads to efficient states much faster than how $\epsilon$-Nash dynamics (i.e., sequences of moves reducing the cost of a player by at least a factor of $\epsilon$) leads to $\epsilon$-Nash equilibria \cite{ChienS07}. Furthermore, also in the more general asymmetric setting, Theorem \ref{ub} shows that the same holds for fair best response dynamics with respect to $\epsilon$-Nash ones \cite{AwerbuchAEMS08}.

Although we have focused on linear congestion games, all the results can be extended to the more general case of congestion games with polynomial delays by exploiting techniques similar to the ones used in \cite{FanelliFM08,FanelliM09}. In particular, it is possible to show that, if each player is allowed to play at least once and at most $\beta$ times any $T$ best responses, states with approximation ratio $O(\beta)$ times the price of anarchy are reached after $T \lceil \log \log n \rceil$ best responses, that such a bound is essentially tight also after exponentially many ones and that, in the symmetric setting,  the game always converges, after $T \lceil \log \log n \rceil$ best responses and for any value of $\beta$, to states with approximation ratio order of the price of anarchy.

{

\end{document}